\documentclass{llncs}
\usepackage[utf8]{inputenc}
\usepackage[english]{babel}
\usepackage[T1]{fontenc}
\usepackage{setspace}
\usepackage{graphicx, caption}
\usepackage{multirow}
\usepackage[dvipsnames, usenames]{color}
\usepackage{amsmath}
\usepackage{subfigure}
\usepackage{amsmath}
\usepackage{amsfonts}
\usepackage{verbatim}
\usepackage{enumitem}
\usepackage{amssymb}
\usepackage{pifont}
\usepackage{sectsty}

\usepackage{multicol}
\usepackage{wrapfig}
\newcommand\blfootnote[1]{%
	\begingroup
	\renewcommand\thefootnote{}\footnote{#1}%
	\addtocounter{footnote}{-1}%
	\endgroup
}

\begin{document}  
\captionsetup{font=small, textfont=it}
\title{Algorithms and Bounds for Drawing Directed Graphs}
\author{Giacomo Ortali \inst{*} \and Ioannis G. Tollis \inst{**}}
\institute{University of Perugia \email{giacomo.ortali@gmail.com} \and  Computer Science Department, University of Crete, Heraklion, Crete, Greece and \\Tom Sawyer Software, Inc. Berkeley, CA 94707 U.S.A.\email{tollis@csd.uoc.gr}}
\maketitle 
\begin{abstract}
	In this paper we present a new approach to visualize directed graphs and their hierarchies that completely departs from the classical four-phase framework of Sugiyama and computes readable hierarchical visualizations that contain the complete reachability information of a graph. Additionally, our approach has the advantage that only the necessary edges are drawn in the drawing, thus reducing the visual complexity of the resulting drawing. Furthermore, most problems involved in our framework require only polynomial time.  Our framework offers a suite of solutions depending upon the requirements, and it consists of only two steps: (a) the cycle removal step (if the graph contains cycles) and (b) the channel decomposition and hierarchical drawing step. Our framework does not introduce any dummy vertices and it keeps the vertices of a channel \emph{vertically aligned}. 
	The time complexity of the main drawing algorithms of our framework is $O(kn)$,  where $k$ is the number of channels, typically much smaller than $n$ (the number of vertices).
\end{abstract}
\section{Introduction}
The visualization of directed (often acyclic) graphs is very important for many applications in several areas of research and business. This is the case because such graphs often represent hierarchical relationships between objects in a structure (the graph).  In their seminal paper of 1981, Sugiyama, Tagawa, and Toda~\cite{DBLP:journals/tsmc/SugiyamaTT81} proposed a four-phase framework for producing hierarchical drawings of directed graphs. This framework is known in the literature as the "Sugiyama" framework, or algorithm. Most problems involved in the optimization of various phases of the Sugiyama framework are NP-hard. In this paper we present a new approach to visualize directed graphs and their hierarchies that completely departs from the classical four-phase framework of Sugiyama and computes readable hierarchical visualizations that contain the complete reachability information of a graph. \blfootnote{* This author's research was performed in part while he was visiting the University of Crete.} Additionally, our approach has the advantage that only the necessary edges are drawn in the drawing, thus reducing the visual complexity of the resulting drawing. Furthermore, most problems involved in our framework require polynomial time.

Let $G=(V,E)$ be a directed graph with $n$ vertices and $m$ edges.
The Sugiyama Framework for producing hierarchical drawings of directed graphs  consists of four main phases~\cite{DBLP:journals/tsmc/SugiyamaTT81}: (a) Cycle Removal, (b) Layer Assignment, (c) Crossing Reduction, and (d) Horizontal Coordinate Assignment.
The reader can find the details of each phase and several proposed algorithms to solve various of their problems and subproblems in Chapter 9 of the Graph Drawing book of \cite{DBLP:books/ph/BattistaETT99}. Other books have also devoted significant portions of their Hierarchical Drawing Algorithms chapters to the description of this framework~\cite{KW,handbook}.

The Sugiyama framework has also been extensively used in practice, as manifested by the fact that various systems have chosen it to implement hierarchical drawing techniques. Several systems such as \emph{AGD} \cite{DBLP:journals/spe/PaulischT90}, \emph{da Vinci} \cite{DBLP:conf/gd/FrohlichW94}, \emph{GraphViz} \cite{DBLP:journals/spe/GansnerN00}, \emph{Graphlet} \cite{DBLP:journals/spe/Himsolt00}, \emph{dot} \cite{gansner2006drawing}, and others implement this framework in order to hierarchically draw directed graphs.  Even commercial software such as the Tom Sawyer Software TS Perspectives \cite{Tom} and yWorks \cite{yWorks} essentially use this framework in order to offer automatic hierarchical visualizations of directed graphs.  More recent information regarding the Sugiyama framework and newer details about various algorithms that solve its problems and subproblems can be found in~~\cite{handbook}.

Even tough this framework is very popular, it has several limitations: as discussed above, most problems and subproblems that are used to optimize the results of each phase have turned out to be NP-hard. Several of the heuristics employed to solve these problems give results that are not bounded by any approximation. Additionally, the required manipulations of the graph often increase substantially the complexity of the graph itself (such as the number of dummy vertices in phase b can be as high as $O(nm)$). The overall time complexity of this framework (depending upon implementation) can be as high as $O((nm)^2)$, or even higher if one chooses algorithms that require exponential time.  Finally, the main limitation of this framework is the fact that the heuristic solutions and decisions that are made during previous phases (e.g., crossing reduction) will influence severely the results obtained in later phases. Nevertheless, previous decisions cannot be changed in order to obtain better results.\\

In this paper we propose a new framework that departs completely from the typical Sugiyama framework and its four phases. Our framework is based on the idea of partitioning the vertices of a graph $G$ into \emph{channels}, that we call \emph{channel decomposition} of $G$.  Namely, after we partition the vertices of $G$ into channels, we compute a new graph $Q$ which is closely related to $G$ and has the same reachability properties as $G$.  The new graph consists of the vertices of $G$, \emph{channels edges} that connect vertices that are in the same channel, and \emph{cross edges} that connect vertices that belong to different channels. Our framework draws either (a) graph $G$ without the transitive "channel edges" or (b) a condensed form of the transitive closure of $G$.  Our idea is  to compute a hierarchical drawing of $Q$ and, since $Q$ has the same reachability properties as $G$, this drawing contains most edges of $G$ and gives us all the reachability information of $G$.  The "missing" incident edges of a vertex can be drawn interactively on demand by placing the mouse on top of the vertex and its incident edges will appear at once in red color. 

Our framework offers a suite of solutions depending upon the requirements of the user, and it consists of only two steps: (a) the cycle removal step (if the graph contains cycles) and (b) the channel decomposition and hierarchical drawing step. Our framework does not introduce any dummy vertices, keeps the vertices of a channel \emph{vertically aligned} and it offers answers to reachability queries between vertices by traversing at most one cross edge.   
Let $k$ be the number of channels  and $m'$ be the number of cross edges in $Q$.  We show that $m'=O(nk)$. The number of bends we introduce is at most $O(m')$ and the required area is at most $O(nk)$. The number of crossings between cross edges and channels can be minimized in $O(k!k^2)$ time, which is reasonable for small $k$.  If $k$ is large, we present linear-time heuristics that find a small number of such crossings.  The total time complexity of the algorithms of our framework is $O(kn)$ plus the time required to compute the channel decomposition of $G$, which depends upon the type of channel decomposition required.

Our paper is organized as follows: the next section presents necessary preliminaries including a brief description of the phases of the Sugiyama framework, the time complexity of the phases, and a description of "bad" choices. In Section 3 we present the concept of path decomposition of a DAG and of path graph (i.e., when the channels are required to be paths of $G$) and we present the new algorithm for hierarchical drawing which is based on any (computed) path decomposition of a DAG. Section 4 presents the concepts of channel decomposition of a DAG and of channel graph (where channels are not paths) and the new algorithm for hierarchical drawing which is based on any (computed) channel decomposition of a DAG.  In Section 5 we present the properties of the drawings obtained by our framework, we offer comparisons with the drawings obtained by traditional techniques, and present our conclusions. Due to space limitations, we present the techniques on minimizing the number of crossings between cross edges and channels in the Appendix.
\section{Sugiyama Framework}
Let $G=(V,E)$ be a directed graph with $n$ vertices and $m$ edges.
A \emph{Hierarchical drawing} of $G$ requires that all edges are drawn in the same direction upward (downward, rightward, or leftward) monotonically. If $G$ contains cycles this is clearly not possible, since in a drawing of the graph some edges have to be oriented backwards. The Sugiyama framework contains the Cycle Removal Phase in which a (small) subset of edges is selected and the direction of these edges is reversed. Since it is important to maintain the character of the input graph, the number of the selected edges has to be minimum.  This is a well known NP-hard problem, called the \emph{Feedback Arc Set} problem. A well known approximation algorithm, called \emph{Greedy-Cycle-Removal}, runs in linear time and produces sets that contain at most $m/2 - n/6$ edges. If the graph is sparse, the result is further reduced to $m/3$ edges~\cite{DBLP:books/ph/BattistaETT99}.\\
\indent
Since the input graph $G$ may contain cycles our framework also needs to remove or absorb them. One approach is to use a cycle removal algorithm (similar to Sugiyama's first step) but instead of reversing the edges, we propose to remove them, since reversing them could lead to an altered transitivity of the original graph.  This is done because the reversed edge will be a transitive edge in the new graph and hence it may affect the drawing.  By the way, this is another disadvantage of Sugiyama's framework.    Since the removal and/or reversal of such edges will create a graph that will have  a "different character" than the original graph we propose another possibility that will work well if the input graphs do not contain long cycles. It is easy to (a) find the \emph{Strongly Connected Components (SCC)} of the graph in linear time, (b) cluster and collapse each SCC into a supernode, and then the resulting graph $G'$ will be acyclic. 
Even if both techniques are acceptable, we believe that the second one might be able to better preserve the character of the input graph.  On the other hand, this technique would not be useful if most vertices of a graph are included in a very long cycle. From now on, we assume that the given graph is acyclic after using either of the techniques described above.\\
\indent
In the Layer Assignment Phase of the Sugiyama framework the vertices are assigned to a layer and the layering is made \emph{proper}, see~\cite{DBLP:books/ph/BattistaETT99,handbook,DBLP:journals/tsmc/SugiyamaTT81}. In other words, long edges that span several layers are broken down into many smaller edges by introducing dummy vertices, so that every edge that starts at a layer terminates at the very next layer. Clearly, in a graph that has a longest path of length $O(n)$ and $O(m)$ transitive edges, the number of dummy vertices can be as high as $O(nm)$. This fact will impact the running time (and space) of the subsequent phases, with heaviest impact on the next phase, Crossing Reduction Phase.\\
\indent
The Crossing Reduction Phase is perhaps the most difficult and most time-consuming phase. It deals with various difficult problems that have attracted a lot of attention both by mathematicians and computer scientists. It is outside the scope of this paper to describe the various techniques for crossing reduction, however, the reader may see~\cite{DBLP:books/ph/BattistaETT99,handbook} for further details. The most popular technique for crossing reduction is the \emph{Layer-by-Layer Sweep}~ \cite{DBLP:books/ph/BattistaETT99,handbook}. This technique solves multiple problems of the well known \emph{Two-Layer-Crossing Problem} by considering the layers in pairs going up (or down).  Of course, a solution for a specific two layer crossing problem "fixes" the relative order of the
vertices (real and dummy) for the next two layer crossing problem, and so on. Therefore, "bad" choices may propagate. Please notice that each two layer crossing problem is NP-complete~\cite{DBLP:journals/algorithmica/EadesW94}. The heuristics employed here tend to reduce crossings by various techniques, but notice that the number of crossings may be as high as $O(M'^2)$, where
$M'$ is the number of edges between the vertices of two adjacent layers.\\
\indent
Finally, in the last phase the exact $x$-coordinates of the vertices are computed by quadratic-programming techniques~\cite{DBLP:books/ph/BattistaETT99,handbook}, which require considerable computational resources.
The dummy vertices are replaced by bends. This implies that the number of bends is about equal to the number of dummy vertices (except when the edge segments are completely aligned).
\section{Path Constrained Hierarchical Drawing}
Let $G=(V,E)$ be a DAG.
In this paper we define a \emph{path decomposition} of $G$ as a set of vertex-disjoint paths $S_p= \{P_1,...,P_k\}$ such that $V(P_1),...,V(P_k)$ is a partition of $V(G)$. A path $P_h\in S_p$ is called a \emph{decomposition path}. The vertices in a decomposition path are clearly ordered in the path, and we denote by $v_i^j$ the fact that $v$ is the jth vertex of path $P_i$. The \emph{path decomposition graph}, or simply path graph, of $G$ associated with path decomposition $S_p$ is a graph $H=(V,A)$ such that $e=(u,v)\in A$ if and only if $e\in E$ and (a) $u,v$ are consecutive in a path of $S_p$ (called \emph{path edges}) or (b) $u$ and $v$ belong to different paths (called \emph{cross edges}).
In other words, an edge of $H$ is a \emph{path edge} if it connects two consecutive vertices of the same decomposition path, else it is a \emph{cross edge}.
Notice that the edges belonging to $G$ but not to $H$ are transitive edges between vertices of the same path of $G$. 

A \emph{path constrained hierarchical drawing} (PCH drawing) $\Gamma$ of $G$ given $S_p$ is a hierarchical drawing of $H$ such that two vertices are drawn on the same vertical line (i.e., same $x$-coordinate) if and only if they belong to the same decomposition path. In this section we propose an algorithm that computes PCH drawings assigning to each vertex the $x$-coordinate of the path it belongs to and for $y$-coordinate we will use its rank in a topological sorting. We will prove that this assignment lets us obtain good results in terms of both area and number of bends.

Next we present Algorithm PCH-Draw that computes a PCH drawing $\Gamma$ of $G$ such that every edge of $G$ bends at most once. We denote by $X(P_h)$ the x-coordinate of Path $P_h$ and by $X(v),Y(v)$ the x-coordinate and the y-coordinate of any vertex $v$. 
Let $P_v$ be the path of $S_p$ containing $v$. By definition of PCH drawing we have that $X(v)=X(P_v)$. Suppose that the vertices of $G$ are topologically ordered and let $T(v)$ be the position of $v$ in a topological order of $V$. PCH-Draw associates to every path, and consequently to every vertex of the path, an x-coordinate that is an even number and to every vertex a y-coordinate that corresponds to its topological order, i.e., $Y(v)=T(v)$ (Steps 1-4). The algorithm draws every edge $e=(u,v)$ as a straight line if the drawn edge doesn't intersect a vertex $w$ different from $u$ and $v$ in $\Gamma$ (Steps 5-7). Otherwise it draws edge $e$ with one bend $b_e$ such that: its x-coordinate $X(b_e)$ is equal to $X(u)+1$  if $X(u)<X(v)$, or $X(u)-1$ if $X(u)>X(v)$. The y-coordinate of bend $b_e$ $Y(b_e)$ is equal to $Y(v)-1$ (Steps 8-14). \\\\
\textbf{Algorithm} PCH-Draw($G=(V,E)$, $S_p=\{P_1, P_2,..., P_k\}$, $H=(V,A)$)\\
1. \indent \textbf{For} $i=1$ to $k$ do\\
2. \indent \indent $X(P_i)=2i$\\
3. \indent \textbf{For} any $v\in V$\\
4. \indent \indent $(X(v),Y(v))=(X(P_v),T(v))$\\
5. \indent \textbf{For} any $e=(u,v)\in A$\\
6. \indent \indent \textbf{If} the straight line drawing of $e$ does not intersect a vertex different \indent \ \ \ \ \ \ \ \ from $u,v$:\\
7. \indent \indent \indent Draw $e$ as a straight line\\
8. \indent \indent \textbf{Else}:\\
9. \indent \indent \ \ \ \ \ \textbf{If} $X(u)<X(v)$:\\
10.\indent \indent \indent \indent $X(b_e)=X(u)+1$\\
11.\indent \indent \indent \textbf{Else}:\\
12.\indent \indent \indent \indent $X(b_e)=X(u)-1$\\
13.\indent \indent \indent $Y(b_e)=Y(v)-1$\\
14.\indent \indent \indent Draw $e$ with one bend at point $(X(b_e),Y(b_e))$\\\\
\noindent
In Figure \ref{figure:9030} we show an example of a drawing computed by Algorithm PCH-Draw. In (a) we show the drawing of a graph $G$ as computed by Tom Sawyer Perspectives (a tool of Tom Sawyer Software) which follows the Sugiyama Framework. In (b) we show the drawing $\Gamma$ of $H$ computed by Algorithm PCH-Draw. The path decomposition that we used to compute the drawing is $S_p=\{P_1,P_2,P_3\}$, where: $P_1=\{0,1,4,7,12,13,15,16,17,20,22,24,25,26,29,30\}$; $P_2=\{2,5,9,11,$ $23,27\}$; $P_3=\{3,6,8,10,14,18,19,21,28\}$. Edge $e=(21,25)$ is the only one bending. In grey we show edge $e$ drawn as straight line, intersecting vertex $23$. 
\begin{figure}[ht]
	\centering
	\subfigure[]
	{\includegraphics[width=0.37\linewidth]{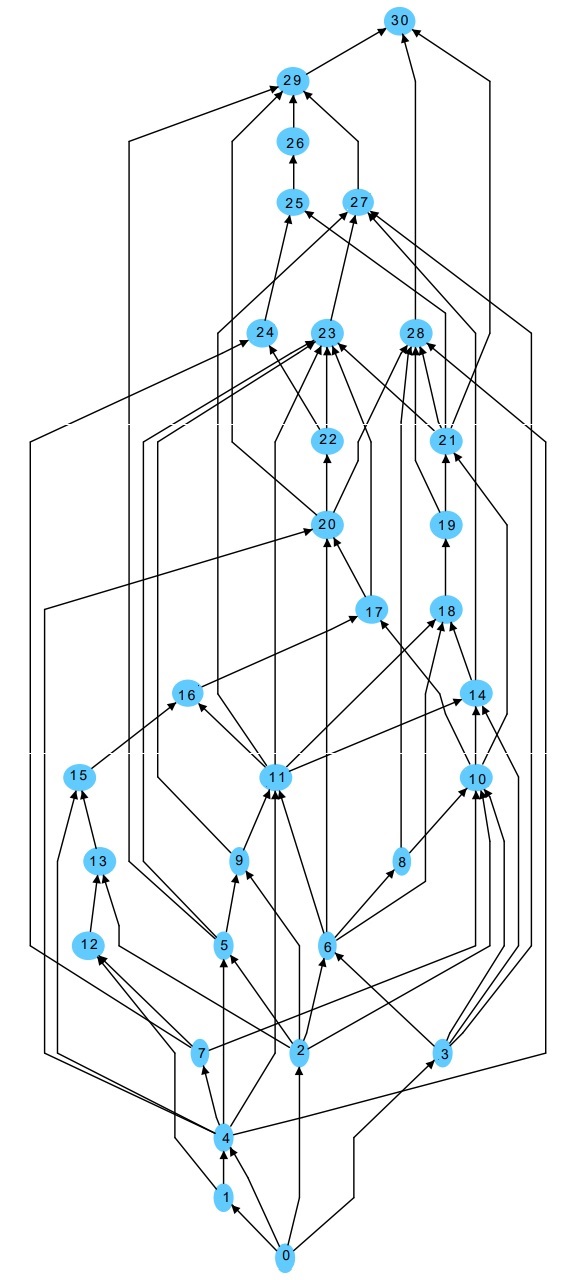}}
	\hspace{15mm}
	\subfigure[]
	{\includegraphics[width=0.27\linewidth]{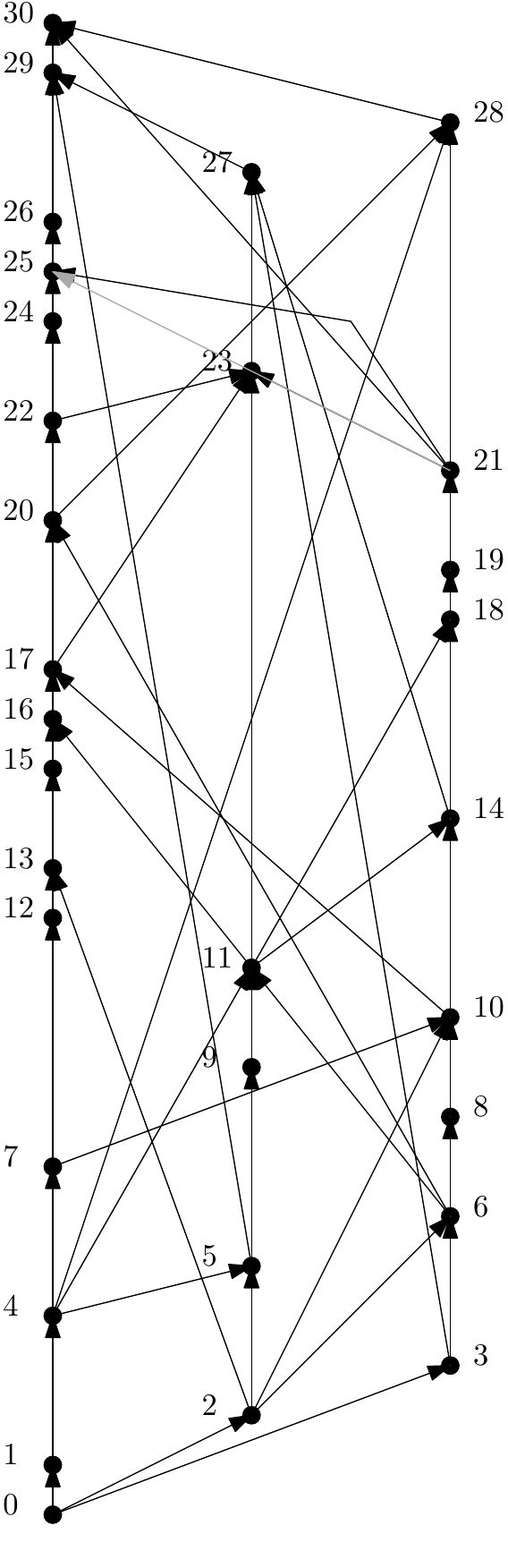}}
	\caption{(a)Drawing of a dag $G$ computed by Tom Sawyer Perspectives (a tool of Tom Sawyer Software) (b) PCH drawing of $H$ computed by Algorithm PCH-draw.}
	\label{figure:9030}
\end{figure}  

Any drawing $\Gamma$ computed by Algorithm PCH-Draw has several interesting properties. First, the area of  $\Gamma$ is typically less than $O(n^2)$. 	By construction, $\Gamma$ has height $n-1$ and width of $2k-1$.
Hence 	$Area(\Gamma)=O(kn)$. 
Given $S_p$ and the topological order of the vertices of $G$, every vertex need O(1) time to be placed. Every edge $e=(u,v)$ needs $O(k)$ time to be placed, since before drawing it we need to check if its straight line drawing would intersect a vertex different from $u,v$ (Step 6). Since the drawing of $e$ must be monotonous, it can intersect at most one edge per path, so we just need to check if in correspondence of every path placed between the path of $u$ and $v$ in $\Gamma$ the drawing of $e$ intersects some vertex.  Hence we have:
\begin{theorem}
	\label{theorem:area}
	Algorithm PCH-Draw computes a drawing 	$\Gamma$ of a DAG $G$ in $O(n+mk)$ time.  Furthermore, 	$Area(\Gamma)=O(kn)$.
\end{theorem}
\noindent
The proofs of Lemma \ref{lemma:no_intersection} and Lemma \ref{figure:overlap} are in the Appendix:
\begin{lemma}
	\label{lemma:no_intersection}
	A cross edge $e=(u,v)$ does not intersect a vertex different from $u$ and $v$ in $\Gamma$.
\end{lemma}
\begin{lemma}
	\label{lemma:overlap}
	Let $e=(u,v)$ and $e'=(u',v')$ be two cross edges drawn with a bend in $\Gamma$. Their bends are placed in the same point if and only if $u$ and $u'$ are in the same decomposition path and $v=v'$.
\end{lemma}
\noindent
In the case described by the above lemma, two edges have overlapping segments $(b_e,v)$ and $(b_{e'},v)$.  We consider this feature acceptable, or even desirable for two edges that have the same endpoint. This typical merging of edges has been used in the past, see for example \cite{Bannister:2015:COD:2951136.2951186,JGAA-391,DBLP:conf/gd/PupyrevNK10}. However, in case that this feature is not desirable, we propose two alternative solutions that avoid this overlap.  The price to pay is larger area, or less edges drawn: 
\begin{enumerate}[leftmargin=*]
	\item  Larger area option:  We can shift horizontally by one unit the position of bend $b_{e'}$ and all the vertices $v$ and bends $b$ such that $X(v)>X(b_e)$ and $X(b)>X(b_e)$. In this case we have no overlaps, but the area of $\Gamma$ can be as large as $O(knm)$.
	\item  Less edges option:  We can define the path decomposition graph differently by removing some transitive cross edges from $H$. For every vertex $v$ we remove the edge $(u,v)$ if there exists an edge $(u',v)$ such that $u'$ and $u$ are in the same decomposition path $P$ and $u$ precedes $u'$ in the order of $P$. It is easy to prove that $H'$ is a subgraph of $H$ and that $A - A'$ contains only transitive edges of $G$. By definition of $H'$, given a decomposition path $P$, for any vertex $v$ there exists at most one cross edge $e=(u,v)$ such that $u\in P$. According to Lemma \ref{lemma:overlap}, there are no bends overlapping. The area of a drawing $\Gamma$ computed using $H'$ is $Area(\Gamma')=O(kn)$. However, we pay for the absence of overlapping bends by the exclusion from the drawing of some transitive cross edges of $G$.
\end{enumerate}
In Figure \ref{figure:overlap} we show an example of the edge overlap described above in a drawing of $H$. Part (a) shows a simple drawing where two edges, $e_1=(u,v)$ and $e_2=(u',v)$, overlap. In grey the drawings of $e_1,e_2$ as straight lines, please notice that both of them intersect a vertex. In part (b) we shift horizontally the drawing, removing the overlap but, of course, increasing the area. Part (c) shows the drawing of $H'$, where edge $(u,v)$ is removed since $u'$ has a higher order in their path.
\begin{figure}[ht]
	\centering
	\subfigure[]
	{\includegraphics[width=0.23\linewidth]{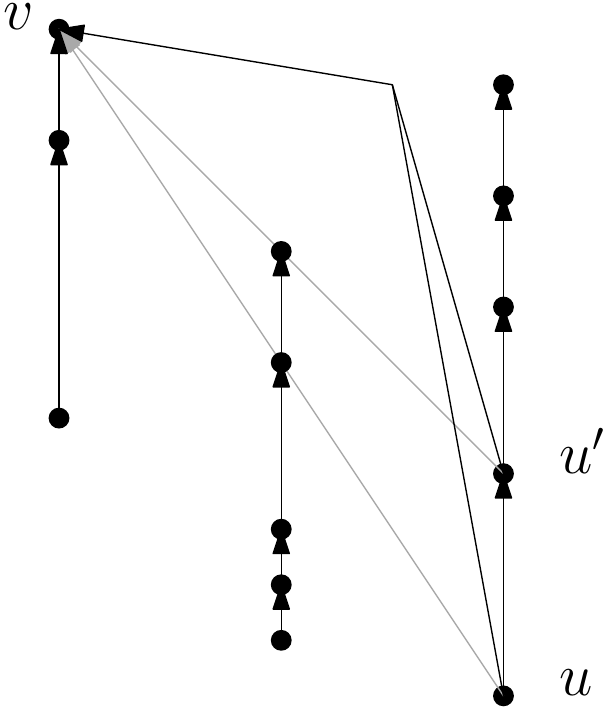}}
	\hspace{15mm}
	\subfigure[]
	{\includegraphics[width=0.26\linewidth]{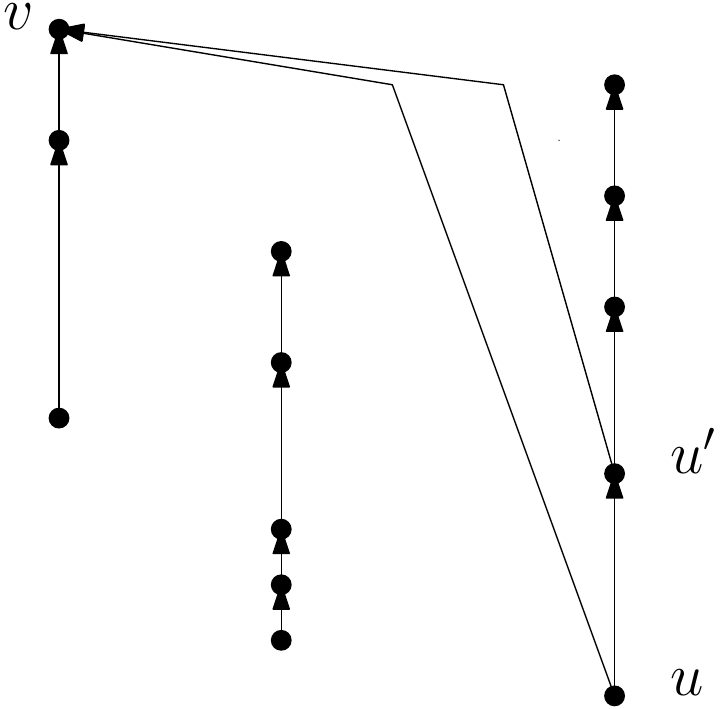}}
	\hspace{15mm}
	\subfigure[]
	{\includegraphics[width=0.23\linewidth]{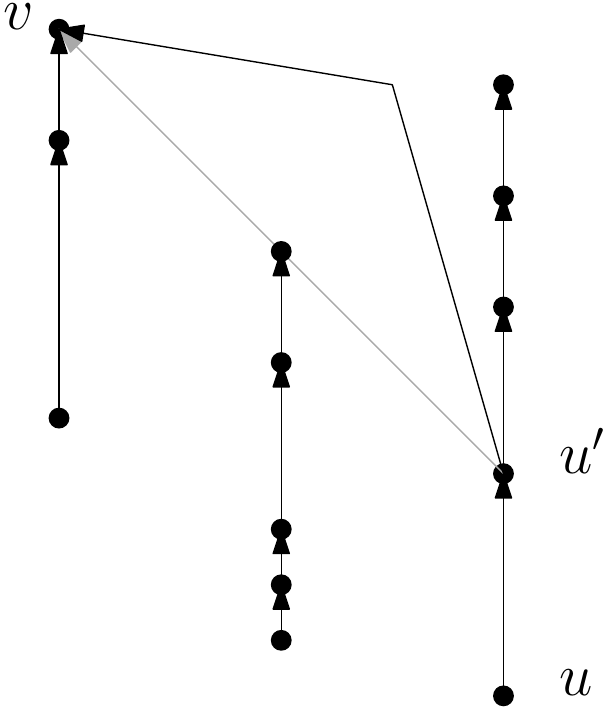}}
	\caption{Examples of bend and edge overlaps in a drawing of $H$.}
	\label{figure:overlap}
\end{figure}

Alternatively we propose to draw every cross edge with a bend. In this case we can avoid Step 6, so we can obtain the drawing in $O(n+m)$. Of course, we pay the reduced time complexity by having more bends in the drawing.

Notice that graphs $H$ and $H'$ are computed from $G$ by simply removing some transitive edges.  Hence we have the following:
\begin{theorem}
	\label{theorem:HH'}
	The path decomposition graphs $H$ and $H'$ have the same reachability properties of $G$.
\end{theorem}
\noindent
Theorem \ref{theorem:HH'} is rather simple, but it is very important, since it tells us that visualizing a hierarchical drawing of $H$ or $H'$ we do read and understand correctly any reachability relation between the vertices of $G$.

A path decomposition $S_p$ of a DAG $G$ with a small number of paths lets us compute a readable PCH drawing of $G$, since the number of decomposition paths influences the area of the drawing and its number of bends. Indeed, since a cross edge can intersect at most one vertex of every decomposition path, the number of decomposition paths influences the number of bends of the drawing.  Furthermore, the number of paths $k$ clearly influences the time to find the minimum number of crossings between cross edges and paths, as it is described in the Appendix. Several algorithms solve the problem of finding a path decomposition of minimum size \cite{DBLP:journals/siamcomp/HopcroftK73,DBLP:conf/recomb/KuosmanenPGCTM18,DBLP:conf/stoc/Orlin13,DBLP:journals/siamcomp/Schnorr78}. The algorithm  of \cite{DBLP:conf/recomb/KuosmanenPGCTM18} is the fastest one for sparse and medium DAGs. In the next section we introduce a relaxed definition of path, the channel, and a way to obtain hierarchical drawings based on a channel decomposition. Notice that, since paths are constrained versions of channels, we expect the minimum size of a channel decomposition to be lower than or in the worst case  equal to the minimum size of a path decomposition. Therefore, we now turn our attention to the concept of a channel decomposition.
\section{Channel Constrained Hierarchical Drawing}
Let $G=(V,E)$ be a DAG. A \emph{channel} $C$ is an ordered set of vertices such that any vertex $u\in C$ has a path to each of its successors in $C$. In other words, given any two vertices $v,w\in C$, $v$ precedes $w$ in the order of channel $C$ if and only if $w$ is reachable from $v$ in $G$. A channel can be seen as a generalization of a path, since a path is always a channel, but a channel may not be a path. A \emph{channel decomposition} $S_c=\{C_1,...,C_k\}$ is a partition of the vertex set $V$ of the graph into channels. If vertex $v$ belongs to channel $C_i$ we write $v_i^j$ if $v$ is the $j$th vertex of channel $C_i$.  The channel decomposition graph $H''$ and a \emph{channel constrained hierarchical drawing} (CCH drawing) of $G$ are defined in a similar fashion as we defined the path decomposition graph $H$ and the PCH drawing of $G$ in the previous section. Notice that, since the channel is a generalization of a path, the concepts of channel decomposition graph and CCH drawing are a generalization of the concepts of path decomposition graph and PCH drawing. We can define Algorithm CCH-Draw in a similar fashion as Algorithm PCH-draw, and its pseudocode is similar to the pseudocode of Algorithm PCH-draw.  The only difference is that Algorithm CCH-Draw takes as input a channel decomposition instead of a path decomposition and that its output is a CCH drawing instead of a PCH drawing. Algorithm CCH-Draw is clearly a generalization of Algorithm PCH-Draw.  Because of space limitations we do not discuss the complete details of Algorithm CCH-Draw here.\\

\indent
Now we introduce a "special" transitive closure, called compressed transitive closure, which is based on the concept of channel decomposition. This transitive closure is obtained from an ordinary transitive closure by removing some of its transitive edges.  Next, we will define a graph $Q$, based on the compressed transitive closure, that will let us obtain more readable drawings.
\\
\noindent
\textbf{Compressed Transitive Closure (CTC):}  Let $L_v$ be a list of vertices associated with a vertex $v\in V$ such that: $L_v$ contains at most one vertex of any decomposition channel; a vertex $w$ is reached from $v$ in $G$ if and only if list $L_v$ contains a vertex $w'$ such that: $w$ and $w'$ are in the same decomposition channel and $w'$ precedes $w$ in the order of that decomposition channel.

The \emph{compressed transitive closure} (CTC) of $G$ is the set of all the lists $L_v$. In \cite{DBLP:journals/tods/Jagadish90} it is shown how to compute the CTC of a graph in $O(mk)$ time. Next we show how we can store the CTC in $O(nk)$ space and that it contains the complete reachability information of $G$.\\
\indent 
We define the \emph{compressed transitive closure graph (CTC graph)} $Q=(V,I)$ such that $(u,v)\in I$ if and only if $u$ is the highest vertex in the order of its channel such that $v\in L_u$. Notice that an edge of $Q$ may not exist in the original graph $G$, as is the case in the ordinary transitive closure graph $G_c$ of $G$.  Furthermore, an edge of $G$ may not be included in $Q$, while $G_c$ contains all the edges of $G$.  Please notice that $Q$ has the same reachability properties (i.e., the same transitive closure) as $G$, since it is computed directly from the CTC of $G$. We denote by \emph{channel edge} an edge of $Q$ connecting two vertices of the same channel, else it is a \emph{cross edge}, similar to the definition of the previous section.

Let $u_i^j$ be a vertex. The list $L_u$ contains by definition the vertex $v_i^{j+1}$, since it is the lowest vertex in the channel $C_i$ reachable from $u$. Hence we have the following property:
\begin{lemma}
	\label{lemma:mc-path}
	$(u,v)\in I$ for any $u_i^j, v_i^{j+1}$.
\end{lemma}
\noindent
Lemma \ref{lemma:mc-path} implies that the channel decomposition $S_c$ of $G$ is a path decomposition of $Q$, so a CCH drawing of $Q$ is essentially also a PCH drawing and hence we can compute it using Algorithm CCH-Draw or Algorithm PCH-Draw since in this case the two algorithms produce the same result. 

In Figure \ref{figure:2060} an example of a CCH drawing of $G$ computed by Algorithm CCH-Draw using $Q$ as an input is shown: Part (a) shows the original graph $G$ drawn as computed by Tom Sawyer Perspectives that uses the Sugiyama Framework. A channel decomposition of this graph is $S_c=\{C_1,C_2,C_3,C_4\}$, where: $C_1=\{0,2,3,7,8,12,15,16,19\}$; $C_2=\{1,4,9,17\}$; $C_3=\{5,10,13,18\}$; $C_4=\{6,11,14\}$. In part (b) we show the drawing of $Q$ as computed by Algorithm CCH-Draw. The dotted edges are edges that do not exist in $G$. Some channel edges are dotted, since a channel may not be a path of $G$.
\begin{figure}[ht]
	\centering
	\subfigure[]
	{\includegraphics[width=0.33\linewidth]{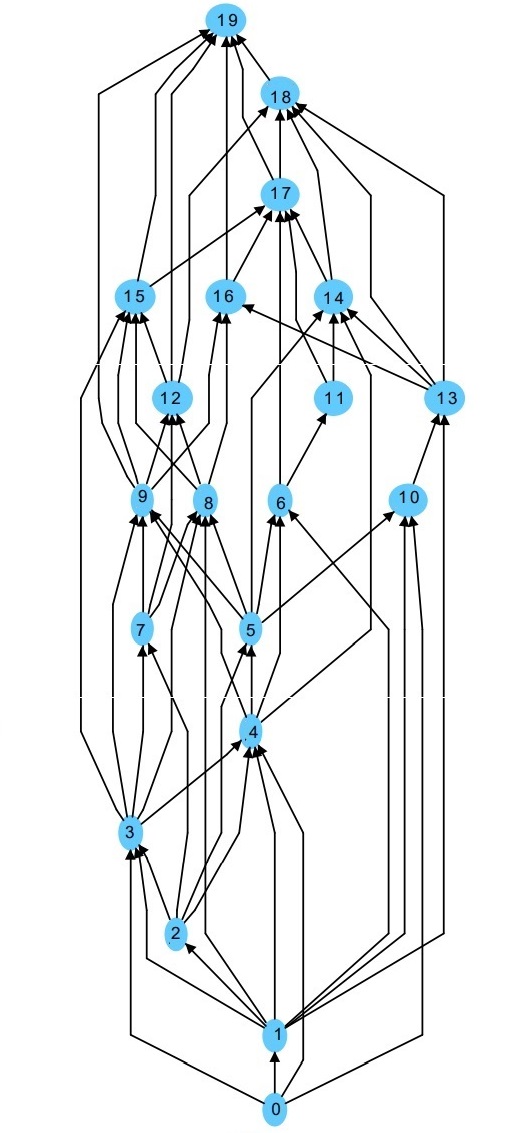}}
	\hspace{15mm}
	\subfigure[]
	{\includegraphics[width=0.29\linewidth]{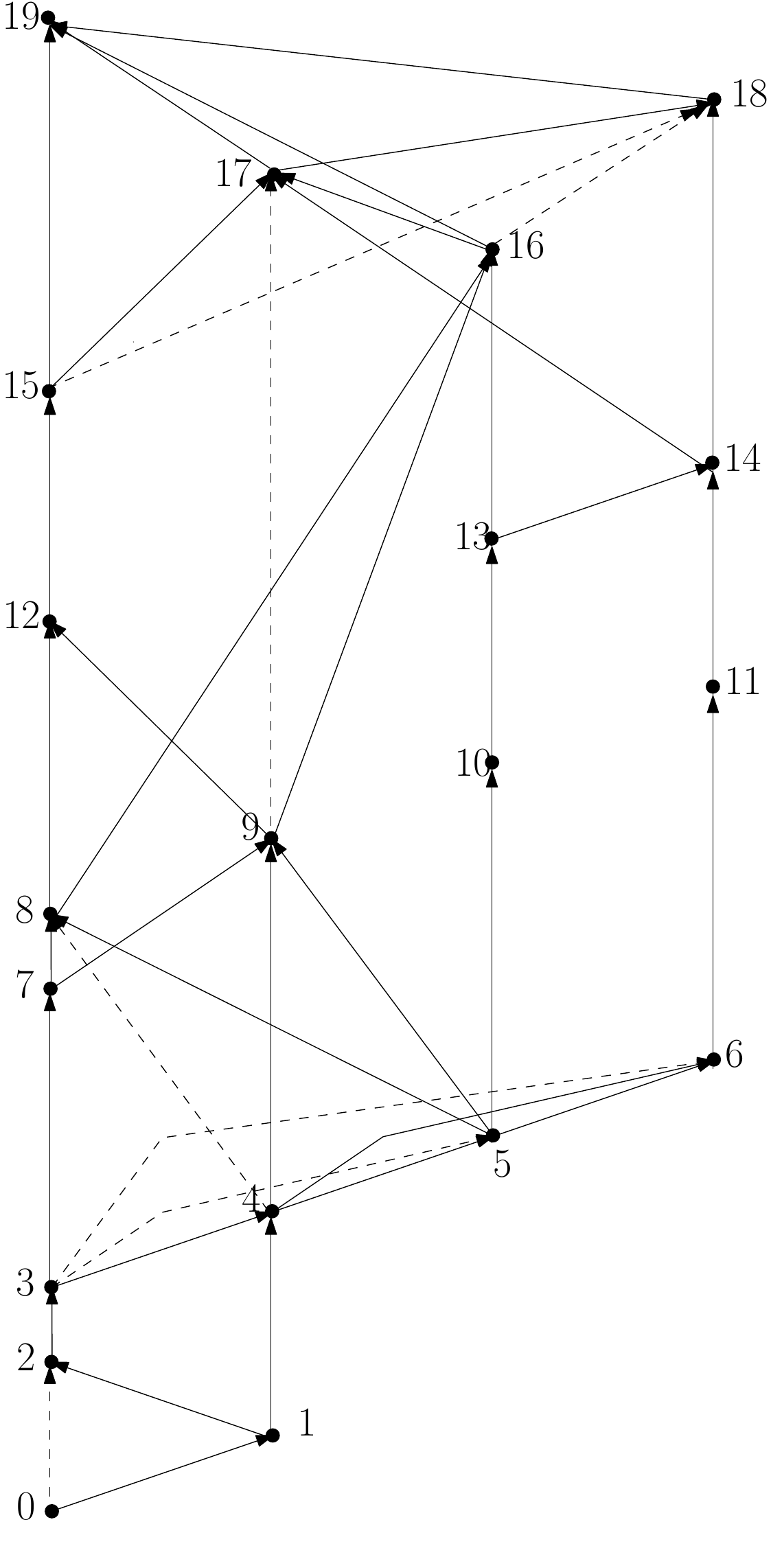}}
	\caption{(a)Drawing of a dag $G$ computed by Tom Sawyer Perspectives (a tool of Tom Sawyer Software) (b) CCH drawing of $Q$ computed by Algorithm CCH-draw.}
	\label{figure:2060}
\end{figure}

There is one list $L_v$ for every vertex $v$ and every list contains $O(k)$ elements.
Since every element of a list $L_v$ corresponds to (at most) one edge of $Q$ we have that $Q$ contains $O(nk)$ edges. Hence we have the following:
\begin{theorem}
	\label{theorem:number}
	The number of edges of $Q$ is $O(nk)$.
\end{theorem}
\noindent 
The above theorem implies that the number of edges of $Q$ is linear if $k$ is a constant. Also, it requires only $O(nk)$ space to be stored.  As we did in the previous section, we denote by \emph{cross edge} an edge of $Q$ connecting two vertices belonging to two different channels. We denote by \emph{mono channel path} (mc-path) a path of $Q$ such that all the edges of it are in the same channel, while we denote by \emph{double channel path} (dc-path) a path of $Q$ composed by two mc-paths and a cross edge. 
\begin{theorem}
	\label{theorem:mc-dc}
	Let $v$ and $u$ be any pair of vertices such that $v$ is reachable from $u$. Then there exists either an mc-path or a dc-path from $u$ to $v$ in $Q$.
\end{theorem}
\begin{proof}
	Suppose that $u$ and $v$ are in the same channel $C_i$. In this case there exists an mc-path from $u$ to $v$ as a consequence of Lemma \ref{lemma:mc-path}. Suppose that $u$ and $v$ are in two different channels $C_i$ and $C_j$. If $u$ reaches $v$, by definition of $Q$, there must be a vertex $u'\in C_i$ that is a successor of $u$ in $C_i$ and a vertex $v'\in C_j$ which is predecessor of $v$ in $C_j$, such that $(u',v')\in I$. Let $p_1$ be the mc-path from $u$ to $u'$ and $p_2$ be the mc-path from $v'$ to $v$. The path $p=p_1+(u',v')+p_2$ is a dc-path from $u$ to $v$.
\end{proof}
\noindent
We claim that such a CCH drawing of $Q$ is a very useful instrument to visualize the reachability properties of $G$. Indeed, if we want to check if a vertex reaches another vertex in $Q$ (and consequently in $G$) we just need to check if there exists an mc-path or a dc-path connecting them (Lemma \ref{theorem:mc-dc}). Moreover, finding an mc-path or a dc-path in $\Gamma$ is very easy, since every mc-path is drawn on a vertical line and every dc-path is drawn as two different vertical lines connected by a cross edge. Moreover since $Q$ has an almost-linear number of edges ($O(nk)$) by Theorem \ref{theorem:number} it makes $Q$ easier to visualize and so it gives us a clear way to visualize the reachability properties of $G$.  The price we have to pay is that we do not visualize many edges of the original graph $G$. These edges can be visualized (in red) on demand by moving the mouse over a given query vertex.

A channel decomposition with a small number of channels lets us compute a readable CCH drawing of $G$. The width $b$ of a DAG $G$ is the maximum cardinality of a subset of $V$ of pairwise incomparable vertices of $G$, i.e., there is no path between any two vertices in the subset.  In \cite{dilworth} it is proved that the minimum value of the cardinality of $S_c$, is $b$ and in \cite{DBLP:journals/tods/Jagadish90} an algorithm is given to compute $S_c$ with $k=b$ in $O(n^3)$ time. The time complexity is improved to $O(bn^2)$ in \cite{chendag}. Clearly, since paths are a restricted type of channels, the minimum size of $S_c$ is less than or equal to the minimum size of $S_p$. 
\section{Comparisons and Conclusions}
We discussed the results of our algorithms in terms of bends and area. The framework we present in this paper produces results that are far superior to the results produced by the Sugiyama framework with respect to the number of crossings, number of bends, area of the drawing and visual clarity of the existing paths and reachability. Namely, because the hierarchical drawings produced by the Sugiyama framework have (a) many crossings (a bound is not possible to be computed), (b) the total number of bends is large and it depends heavily on the number of dummy vertices introduced, (c) the area is large because the width of the drawing is negatively influenced by the number of dummy vertices, (d) the number of bends per edge is also influenced by the number of dummy vertices on it (although the last phase tries to straighten the edges by aligning its segments, at the expense of the area, of course), (e) most problems and subproblems of each phase are NP-hard, and many of the heuristic are very time consuming, and (f) the reachability information in the graph is not easy to detect from the drawing.\\
\indent 
Our framework produces hierarchical drawings that are far superior of the ones produced by the Sugiyama framework in all measures discussed above.  Namely, our drawings have (a) a minimum number of channel crossings as an upper bound (see the Appendix), (b) the total number of bends is low since we introduce at most one bend for some (not all) cross edges, (c) the area is precisely bounded by a rectangle of height $n-1$ and width $O(k)$, where $k$ is typically a small fraction of $n$, (d) the reachability and path information is easily visible in our drawings since any path is deduced by following at most one cross edge (which might have at most one bend), (e) the vertices in each channel are vertically aligned and there is a path from each vertex
in the channel to all the vertices that are at higher $y$-coordinates, (f) all our algorithms run in polynomial time (with the exception of the minimization of the number of channel crossings, which requires $O(k! k^2)$ time), and finally, (g) the flexibility of our framework allows a user to decide to have their specified paths as channels, thus allowing for user paths to be drawn aligned.\\
\indent
The only drawback of the drawings produced by our framework is the fact that it does not draw all the edges of the graph, which might be important for some applications.  This might be considered as an advantage by some other users since it offers drawings that are not cluttered by the edges. In any case, we offer the  remedy to visualize all the edges incident to a vertex interactively when the mouse is placed on top of a vertex.\\
\indent
We believe that the above comparison is convincing of the power of the new framework. Hence we do not offer experimental results here.  However, in the future we plan to contact user studies in order to verify that the users we will benefit  from the aforementioned properties by showing higher understanding and ease of use of the new drawing framework. We plan to work on allowing to include user specified channels (or paths), and still find the minimum number of channels in a channel decomposition. It would be interesting to find specific topological orderings and/or sophisticated layer assignment that will reduce the hight, the number of crossings and the number of bends  of the computed drawing.  Finally, it would be desirable to avoid the exponential in $k$ (i.e., $k!$) factor in the time complexity of finding the best order of the channels.
	\bibliography{literature}
	\chapter*{Appendix}
\subsection*{Proof of Lemma \ref{lemma:no_intersection}}
If $e$ is drawn as a straight line the lemma is true by the construction of Algorithm PCH-Draw. Otherwise, $e$ is composed of two segments: $(u,b_e)$ and $(b_e,v)$. Both segments are diagonals of the rectangles inside which there is no vertices, so the two segments, and consequently $e$, cannot intersect any vertex different from $u,v$.
\subsection*{Proof of Lemma \ref{lemma:overlap}}
$X(u)=X(u')$ since $X(b_e)=X(b_{e'})=X(u)\pm 1=X(u')\pm 1$. In this case $u$ and $u'$ are in the same path of the path decomposition. $v=v'$ since $Y(b_e)=Y(b_{e'})=Y(v)-1=Y(v)-1$ and since there is no vertex $w\not =v$ such that $Y(w)=Y(v)$.
\section*{Minimizing the Number of Crossings}
\label{section:crossings}
We denote by \emph{channel crossing} a crossing between a channel (edge) and a cross edge. In this section we discuss how we can reduce the number of channel crossings of a CCH drawing computed by Algorithm CCH-Draw, by changing the left-to-right order of the channels. Notice that because of the similarities of CCH and PCH drawings, the techniques we will describe for CCH drawings are applicable to the PCH drawings as well.  We will discuss the results with respect to the compressed transitive closure graph $Q=(V,I)$, but the same results can be obtained with any decomposition graph $H$.

Let $X(C)$ be the x-coordinate of the vertices of a channel $C$ in the CCH drawing $\Gamma$. Let $C_i$ and $C_j$ be two different channels. Let $k_{ij}$ be the number of channels that lie between $C_i$ and $C_j$, e.g., $X(C_i)<X(C)<X(C_j)$. Let $I_{ij}$ be the set of the cross edges connecting any two vertices of $C_i$ and$C_j$, where $|I_{ij}|=m_{ij}$. Clearly, every edge of $I_{ij}$ can be involved in at most $k_{ij}$ channel crossings, since this is the number of channels between its beginpoint and its endpoint.  Hence, $m_{ij}k_{ij}$ is the maximum number of crossings of the set of edges $I_{ij}$. Since the set of all possible $I_{ij}$ is a partition of $I$ then we have the following:
\begin{lemma}
	$\theta=\sum_{C_i,C_j\in (S_c,S_c)} m_{ij}k_{ij}$ is an upper bound on the number of channel crossings of a drawing $\Gamma$.
\end{lemma}
\noindent
In order to find the optimum order of $k$ channels that minimizes the upper bound on channel crossings we use a brute force approach. In other words, we simply count the crossings for each of the $k!$ permutations.
Algorithm Best-Order will compute the order of the channels of a drawing $\Gamma$ such that it has a minimum upper bound on the number of crossings, $\theta$. The algorithm takes $S_c$ and $Q$ as input and gives as output an ordered channel decomposition with the order described above. The first step is computing the value $m_{ij}$ for every possible couple of channels of $S_c$. Let $ord$ be an order of the channels of $S_c$, where $ord(C_i)<ord(C_j)$ if $C_i$ precedes $C_j$ in $ord$. Then the algorithm will try all the possible permutations of the channels and for every permutation computes the values $k_{ij}$ for every pair of channels. Then it computes the value $\theta$. Finally the algorithm chooses the order $ord_m$ of the channels with minimum $\theta$ ($\theta_m$) and gives as output the ordered channel decomposition $S_c'$ where the channels are ordered as in $ord_m$.\\

\textbf{Algorithm} Best-Order($Q=(V,I)$,$S_c=\{C_1,...,C_k\}$)\\
1. \indent $m_{ij}=0$ $\forall i,j$\\
2. \indent \textbf{For} any cross edge $e=(u,v)\in I$\\
3. \indent \indent Let $u\in C_i$ and $v\in C_j$\\
4. \indent \indent $m_{ij}++$\\
5. \indent $\theta_m=|E|k$\\
6. \indent $ord_m=void$\\
7. \indent \textbf{For} any possible order $ord$ of $S_c$:\\
8. \indent \indent \textbf{For} any couple of different channels $(C_i,C_j)$:\\
9. \indent \indent \indent $k_{ij}=|ord(C_i)-ord(C_j)|-1$\\
10.\indent \indent $\theta=\sum_{C_i,C_j\in (S_c,S_c)} m_{ij}k_{ij}$\\
11.\indent \indent \textbf{If} $\theta<\theta_m$:\\
12.\indent \indent \indent $\theta_m=\theta$\\
13.\indent \indent \indent $ord_m=ord$\\
14.\indent $S_c':=$ ordered channel decomposition containing the channels of $S_c$\\
15.\indent Order the channels of $S_c'$ as in $ord_m$\\
16.\indent \emph{output} $S_c'$\\

Algorithm Best-Order tries all possible orders of channel decomposition $S_c$ and chooses the one that implies the minimum $\theta$.  Furthermore, since Step 2 requires $O(|I|)$ time and Steps 7 and 8 require $O(k!k^2)$, we have the following:
\begin{theorem}
	$S_c'=$ Best-Order($Q$, $S_c$) is the ordered channel decomposition such that $\Gamma=$ CCH-Draw($S_c'$, $Q$) has minimum $\theta$ among all the channels decomposition having the same channels of $S_c$. Furthermore, it runs in $O(k!k^2+|I|)$ time.
\end{theorem}

\noindent
If $k$ is a small number, Algorithm Best-Order is a nice heuristic to reduce the number of channel crossings, since it minimizes the worst case. Suppose on the other hand that $k$ is large. Let $C$ be a channel and let $E(C)$ be the set of cross edges adjacent to it. Suppose that $C$ is placed in the center of the drawing. The number of channel edges crossing an edge of $E(C)$ is at most $k/2$, since it can cross all the channels on the left or on the right of $C$, and in both cases the number of such channels is about $k/2$. In this case the number of channel crossings inolving edges of $E(C)$ is at most $|E(C)|k/2$. On the other hand, if $C$ is placed in one of the two borders of the drawing then an edge of $E(C)$ can cross all the channels of the drawing, in the case that the other endpoint of that edge is in the channel placed in the other border. In this case the number of channel crossings inolving edges of $E(C)$ is at most $|E(C)|k$. As a conclusion we suggest to place the two channels with the least cross edge degree at the border of the drawing. Next, from the remaining channels, we pick the two channels with the least cross edge degree and place them next to the two channels placed recently toward the center, and so on. Consequently the channels with large $|E(C)|$ will be placed in the center of the drawing.

\end{document}